\documentclass{article}
\usepackage{arxiv}

\usepackage{graphicx}

\usepackage{apxproof}

\usepackage{comment}

\usepackage[utf8]{inputenc} 
\usepackage[T1]{fontenc}    
\usepackage{hyperref}       
\usepackage{url}            
\usepackage{booktabs}       
\usepackage{amsfonts}       
\usepackage{nicefrac}       
\usepackage{microtype}      
\usepackage{lipsum}		
\usepackage{graphicx}
\usepackage{natbib}
\usepackage{doi}

\newtheorem{theorem}{Theorem}
\newtheorem{example}{Example}
\newtheorem{definition}{Definition}
\newtheorem{lemma}{Lemma}
\newtheorem{proposition}{Proposition}
\newtheorem{corollary}{Corollary}

\newcommand{\N}{F(S)}

\title{Complete Trigger Selection in Satisfiability modulo First-order Theories}

\author{ \href{https://orcid.org/0000-0003-1141-0665}{\includegraphics[scale=0.06]{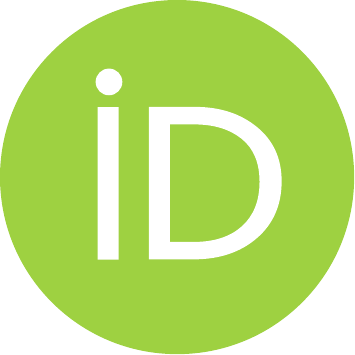}\hspace{1mm}Christopher Lynch} \\
	Department of Computer Science\\
	Clarkson University\\
    8 Clarkson Avenue\\
	Potsdam, NY 13699-5815 \\
	\texttt{clynch@clarkson.edu} \\
	\And
	\href{https://orcid.org/0009-0001-8679-036X}{\includegraphics[scale=0.06]{orcid.pdf}\hspace{1mm}Stephen Miner} \\
	Department of Computer Science\\
	Clarkson University\\
    8 Clarkson Avenue\\
	Potsdam, NY 13699-5815 \\
	\texttt{minersj@clarkson.edu} \\
}

\renewcommand{\shorttitle}{\textit{arXiv} Template}

\begin{document}
\maketitle



\begin{abstract}
Let T be an SMT solver with no theory solvers except for Quantifier Instantiation.  Given a set of first-order clauses S saturated by Resolution (with a valid literal selection function) we show that T is complete if its Trigger function is the same as the literal selection function.  So if T halts with a ground model G, then G can be extended to a model in the theory of S.  In addition for a suitable ordering, if all maximal literals are selected in each clause, then T will halt on G, so it is a decision procedure for the theory S.  Also, for a suitable ordering, if all clauses are Horn, or all clauses are 2SAT, then T solves the theory S in polynomial time. 
\end{abstract}

\section{Introduction}
SMT solvers \cite{DBLP:series/faia/BarrettSST21} are very efficient at satisfiability problems over several theories with specialized decision procedures.  For first-order theories where a specialized decision procedure has not been implemented, a background theory can often be represented by quantified first-order clauses.\footnote{In this paper, the word "theory" refers to a satisfiable set of first-order clauses.}  
The SMT solver instantiates universally quantified first-order clauses into ground clauses, which can be handled by its SAT solver.  To decide which instances are useful the SMT solver can use a process called {\em triggers} \cite{DBLP:journals/jacm/DetlefsNS05,DBLP:conf/cade/MouraB07}.  A trigger function maps each first-order clause to a set of terms in the clause.  If the terms in this set match existing ground terms, that triggers an instantiation.      


Researchers have studied practical methods of selecting triggers.  If triggers are selected well, the SMT solver can quickly solve unsatisfiable problems.
However, if the problem is satisfiable, the SMT solver will often run forever or halt with a partial propositional model. If the SMT solver halts with a partial propositional model, it will not know if that propositional model will extend to a model of the first-order clauses.  

This paper is a result of our initial efforts to understand in what instances an SMT solver can be assured that enough instances have been generated to determine satisfiability.
The subject  is first-order logic without equality.  We are motivated by completeness results involving selection functions in resolution-based first-order theorem proving \cite{DBLP:books/el/RV01/BachmairG01}.  We show a relationship between  selection functions and the trigger functions of SMT solvers.
We have started to extend these results to equational logic \cite{DBLP:books/el/RV01/NieuwenhuisR01},
and future research will be to extend them to specialized theories.

As an example of the problem faced by SMT solvers, consider the following first-order theory represented by clauses, where capital letters are universally quantified variables.  This example shows that even if the first-order theory has no disjunction, SMT solvers already have trouble:

\begin{example}
\label{unit}
$$ g(s(X),X) $$
$$ \neg g(X,X) $$
\end{example}

If we give this theory to z3 \cite{DBLP:conf/tacas/MouraB08} and assert $g(a,b)$, z3 returns "unknown" when using the default mbqi (model-based quantifier instantiation) \cite{DBLP:conf/cav/GeM09}.  If we turn off mbqi and set $g(s(X),X)$ and $g(X,X)$ as triggers, z3 will quickly halt and say "unknown".\footnote{We don't mean to pick on z3.  We also ran this on cvc5 \cite{DBLP:conf/tacas/BarbosaBBKLMMMN22}, veriT \cite{DBLP:conf/cade/BoutonODF09} and SMTInterpol \cite{DBLP:conf/spin/ChristHN12}.  They all returned "unknown" or ran forever.}.  The SMT solver will have generated enough instances to determine satisfiability, but it is not aware of that.  SMT solvers do well with conjunctive normal form problems without uninterpreted function symbols, but may have trouble with satisfiable problems with uninterpreted function symbols.

 We now consider a first-order theory that contains disjunction, 
 to discuss trigger selection:

\begin{example}
\label{usedlater}
$$ C_1:  \neg p(X_1,Y_1) \vee q(f(X_1), Y_1) $$
$$ C_2:  \neg q(X_2,Y_2) \vee p(X_2,f(Y_2)) $$
\end{example}

This is another theory that z3 cannot solve when presented with ground clause $p(a,b)$.  The previous example only consisted of unit clauses, so there was no question of which literals to select for triggers.  But in this example, we need to decide which literals to select for triggers.  So we now consider three possible trigger selection strategies.

\begin{enumerate}
    \item If we select $q(f(X_1), Y_1)$ and $p(X_2, f(Y_2))$ as triggers,  we show whenever an SMT solver halts without saying "unsatisfiable", the ground model it has created is actually a model modulo the first-order theory.  In fact, for a fragment of first-order logic to which this theory with this selection function belongs, we show that the SMT solver is a polynomial-time decision procedure.
    \item If we select $\neg p(X_1,Y_1)$ and $\neg q(X_2,Y_2)$ as triggers,\footnote{To reduce instantiation, we use entire literals as triggers.} a halting SMT solver can determine satisfiability.  Unfortunately, given ground clause $p(a,b)$ the procedure will not halt.
    \item If we select $q(f(X_1),Y_1)$ and $\neg q(X_2,Y_2)$, then the SMT solver will return "unknown", because ground clauses $p(a,b)$ and $\neg p(f(a),f(b))$ are unsatisfiable in that theory, but the SMT solver will not generate any instances. However, if we add the first-order clause $\neg p(X_3,Y_3) \vee p(f(X_3),f(Y_3))$ to the first-order theory, and select either literal in that clause, the SMT solver is complete.
\end{enumerate}

In the first trigger selection (and the third one with the extra clause), the SMT solver creates enough instantiations to guarantee satisfiability.  
Unfortunately, the instantiations will not halt for the second one.
For all but the second case, there is an ordering where we selected all the maximal literals in each clause.  For the second case, there is no such ordering.  

Our purpose is not to create a new inference system, but to understand when existing SMT solvers could answer "satisfiable" instead of "unknown", such as the above examples. If a set of first-order clauses is saturated under Resolution with a valid selection function (as defined below), we choose the literals selected during Resolution to be the triggers. If an SMT solver halts with "unsatisfiable", the problem is unsatisfiable. But if the SMT solver halts without detecting unsatisfiability, most SMT solvers would say "unknown". However, using our method of trigger selection, we can know the problem is satisfiable. Furthermore, the partial ground model that the SMT solver has constructed can be extended to a model of the first-order theory.  In Example~\ref{unit}, the first-order theory is saturated.  In Example~\ref{usedlater}, the first-order theory is saturated under all three selection functions, assuming that the additional clause is added in the third case.

If, in addition, a single maximum literal is selected in each clause in the saturation of the first-order theory, the SMT procedure will halt, and therefore the SMT procedure is a decision procedure.  Alternatively, if the order is isomorphic to $\omega$,
  \footnote{There are only finitely many atoms smaller than any given atom.} 
  then selecting all maximal literals will give a decision procedure.  
  If, in addition, the chosen order is a polynomial ordering which is totalizable on ground terms, the SMT solver is guaranteed to decide satisfiability in polynomial time if all clauses are Horn or all clauses contain at most two literals.

In Section 2 of this paper, we give some well-known definitions and some definitions specific to this paper.  In Section 3 we define the inference rules used to model our procedure.  Section 4 proves the completeness.  Section 5 shows cases where we are guaranteed to have a decision procedure and where it is guaranteed to run in polynomial time.  Section 6 gives related work, and Section 7 summarizes
the paper and gives some important future work.

\section{Preliminaries}
We consider a set of ground formulas modulo a set of first-order formulas, which are in conjunctive normal form.
We follow standard definitions for Resolution 
theorem proving \cite{DBLP:books/el/RV01/BachmairG01}, plus some new definitions that are specific to this paper.

We assume we are given a set of variables, which we represent with capital letters, 
and a
set of uninterpreted function symbols of various arities, represented with lower case letters.  An
arity is a non-negative integer.  {\em Terms} are defined
recursively in the following way:  each variable is
a term, and if $t_1,\cdots,t_n$ are terms, and $f$
is of arity $n \geq 0$, then $f(t_1,\cdots,t_n)$ is a term. 
If $P$ is a predicate symbol of arity $n$, and if 
$t_1,\cdots,t_n$ are terms, then $P(t_1,\cdots,t_n)$
is an {\em atom}.  Any atom or negation of an atom
is a {\em literal}.  A literal is called {\em negative}
if it is negated, and {\em positive} otherwise.  
For all literals $L$, we define 
$\bar{L}$ so that
$\bar{L} = \neg L$ and $\bar{\neg L} = L$.
A {\em clause} is a multiset of literals, representing a disjunction of literals.  If $L$ is a literal, and $\Gamma$ is a set of literals, we will write $L \vee \Gamma$ to represent $\{L\} \cup \Gamma$.
We use $\bot$ to represent the empty clause.
We will use ``$-$'' to denote
multiset difference.
For any object $C$, define $Vars(C)$ as the set of variables in $C$.  If $Vars(C) = \emptyset$ we say that $C$ is 
{\em ground}, otherwise we say that $C$ is non-ground. 

A {\em substitution} is a mapping from the set of 
variables to the set of terms, which is 
almost everywhere the identity.  
We identify a substitution with its homomorphic
extension.  Composition of substitutions $\sigma$ and $\rho$ is defined so that $X(\sigma\rho) = (X\sigma)\rho$ for all variables $X$.  
If $\theta$ is a substitution then $Dom(\theta)
= \{X \,\,|\,\, X\theta \not= X\}$, and
$Ran(\theta) = \{X\theta \, \, | \, \, X \in Dom(\theta)\}$.
A substitution $\theta$ matches $A$ to $B$ if $A\theta = B$, and
is a {\em unifier} of
$A$ and $B$, if $A\theta = B\theta$.
$\sigma$ is a {\em most general unifier} of $A$ and $B$, written
$\sigma = mgu(A,B)$ if $\sigma$ is a unifier of $A$ and $B$,
and for all unifiers $\theta$ of $A$ and $B$, there is a 
substitution $\rho$ such that $X\sigma\rho = X\theta$ for all $X$ in $Vars(A \cup B)$.
Given a clause $C$, define $Gr(C) = \{C\theta | C\theta$ is
ground $\}$.  Given a set of clauses $S$, let $Gr(S) = 
\bigcup_{C \in S} Gr(C)$. 

We assume an ordering $<$ is a well-founded ordering which is {\em stable},
meaning that if $s < t$ then $s\theta < t\theta$.
We assume the ordering is totalizable on all ground terms and atoms.  This means the ordering can be extended to an ordering that is total on ground terms.
It can be extended to literals in
any way such that $A < \neg A$ for all atoms $A$.  We also assume the ordering is an {\em atom ordering} meaning that for all literals $L$ and $M$, $L > M$ implies $L > \bar{M}$.
Clauses are compared
using the multiset ordering.  
A literal $L$ is said to be
{\em maximum} in a clause $C$ if $L$ is larger than all other literals
in $C$, and {\em maximal} in $C$ if no other literal
in $C$ is larger than $L$.  An order is a {\em polynomial ordering} if each atom only has polynomially many smaller atoms.  


An {\em partial interpretation} (or just {\em interpretation}) $I$ is defined as a consistent set of ground literals such that $I \models L$ if and only if $L \in I$.  Therefore, an atom $A$ is undefined in $I$ if $A \not\in I$ and $\neg A \not\in I$.  This differs with some definitions of interpretations where just the true positive literals are given.
Since clauses are multisets representing disjunctions, $I \models C$ if $I \cap C \not= \emptyset$, otherwise $C$ is either false or undefined in $I$.  
If $C$ is not ground then we say $I \models C$ if $I \models Gr(C)$.
If $I$ is an interpretation and $S$ is a set of clauses, then $I$
is a {\em model} of $S$ if $I \models C$ for all $C \in S$. 
Interpretations $I_1$ and $I_2$ are {\em compatible} if there is no literal $L$ such that $L \in I_1$ and $\bar{L} \in I_2$.  If $I_1$ and $I_2$ are compatible then $I_1 \cup I_2$ is also an interpretation, furthermore, for any literal $L$, $I_1 \cup I_2 \models L$ if and only if $I_1 \models L$ or $I_2 \models L$.

Given an interpretation $I$ and a ground clause $C$, let $Filter(C,I) = \{ L \in C \,\,|\,\, I \not\models\bar{L} \}$.  If $S$ is a set of ground clauses, let $Filter(S,I) = \{ Filter(C) \,\,|\,\, C \in S \, I\not\models C\}$.  i.e., $Filter(S,I)$ 
is created from $S$ by removing all clauses true in $I$, and then removing all literals false in $I$ from the remaining clauses.

\begin{example}
Consider interpretation $I = \{\neg p(a), p(b)\}$ where $S$ is the set of clauses in the following example:
    $$C_1: p(a) \vee \neg p(b) \vee p(c)$$
    $$C_2: \neg p(a) \vee \neg p(b) \vee p(d)$$
\end{example}

Then $Filter(C_1,I) = p(c)$,  
and $Filter(S,I)$
is the set consisting of the unit clause $p(c)$.


\section{Inference System}
We want to model an SMT solver without any theories except for a quantified first-order theory in CNF, represented by clauses with universal variables.  Given a set of clauses $S$, let $g(S)$ be the set of all ground clauses in $S$, and let $ng(S)$ be the set of all non-ground clauses in $S$.
An SMT solver would build a model from $g(S)$. 
Call that model $M_{g(S)}$.  For most of this paper, it will not be important how that model is built.  
%
%
%
%
%
The SMT solver will use $M_{g(S)}$ to instantiate the clauses of $ng(S)$.  We will use inference rules to model the instantiation process.

Let $Trig$ be a function so that for each clause $C$ in $ng(S)$,
$Trig(C)\subseteq C$ and $Vars(Trig(C)) = Vars(C)$, which
determines which parts of $C$ are used for instantiation.  
Below we show how to choose triggers in such a way that when we have a model, and no more instantiations can be performed, we can deduce that $S$ is satisfiable.  
The Instantiation rule is used to instantiate non-ground clauses based on a ground interpretation $I$. 

\vspace{0.1cm}

\textbf{$I$-Instantiation:} \[\frac{L_1 \vee \cdots \vee L_n \lor \Gamma}{(L_1\vee \cdots \vee L_n \vee \Gamma)\theta}\] \\
where 
\begin{enumerate}
    \item $L_1 \vee \cdots \vee L_n \vee \Gamma \in ng(S)$,
    \item $Trig(L_1 \vee \cdots \vee L_n \vee \Gamma) = \{L_1, \cdots , L_n\}$
     \item there exists ${L'_1} \cdots {L'_n}$ in $I$ such that $\bar{L_i}\theta = {L'_i}$ for all $1 \leq i \leq n$
\end{enumerate}

We do not consider equality, so 
we only require $\theta$ to be a matcher, not an $E$-matcher.  
SMT solvers allow triggers to be subterms of a literal.  
But to reduce the number of instantiations,
we only use entire literals as triggers.  Furthermore, we only need to match a ground literal in the model onto the complement of a non-ground literal.  This allows the instantiation rule to be more restrictive than is usually the case for trigger-based instantiation in SMT.  Finally, SMT solvers allow for different possible sets of triggers for the same clause.  We only require one set of triggers for each clause.

We say that a set of clauses $S$ is {\em saturated by Instantiation} if either $g(S)$ is unsatisfiable or there exists some model $M_{g(S)}$ of $g(S)$ such that every conclusion of an $M_{g(S)}$-Instantiation inference is in $S$.  

\begin{example}
\label{goodsel}
$$ C_1:  \neg p(X_1,Y_1) \vee q(f(X_1), Y_1) $$
$$ C_2:  \neg q(X_2,Y_2) \vee p(X_2,f(Y_2)) $$
$$ C_3: \neg p(f(a),f(b)) $$
\end{example}

We define $Trig_1$ so that 
$Trig_1(C_1) = \{q(f(X_1),Y_1)\}$ and $Trig_1(C_2) = \{p(X_2,f(Y_2))\}$.
Given the model $M_1 = \{\neg p(f(a),f(b))\}$ of $C_3$, we apply the 
Instantiation rule to create the clause $C_4 = \neg q(f(a),b) \vee
p(f(a),f(b))$. Then, we can create a new model $M_2 = \{\neg p(f(a),f(b)), 
\neg q(f(a),b)\}$ of $C_3$ and $C_4$.  Instantiation then creates 
$C_5 = \neg p(a,b) \vee q(f(a),b)$, and we create a new 
model $M_3 = \{\neg p(f(a),f(b)), \neg q(f(a),b), \neg p(a,b)\}$ of $C_3$ and $C_4$.  
These five clauses are now saturated by Instantiation.

Now consider the same set of three clauses with a different trigger function $Trig_2$ defined so that $Trig_2(C_1) = \{q(f(X_1),Y_1)\}$ and $Trig_2(C_2) = \{\neg q(X_2,Y_2)\}$.  Let $M_1$ be the same model as before.  Then there are no instantiations, so the three clauses are saturated by Instantiation.  Further, let's suppose we also had the clause $C_6 = p(a,b)$.  Then the model of the clauses would be $M_4 = \{\neg p(f(a),f(b)), p(a,b)\}$.  The set of clauses $\{ C_1, C_2, C_3, C_6\}$ is again saturated by Instantiation, even though it is unsatisfiable.  In other words, this was not a good choice of triggers.



Consider the theory 
$\{p(X_1), \neg p(X_2)\}$, with no ground clauses.  This is unsatisfiable, but no instantiations exist. So there cannot be a set of triggers that guarantees completeness for all formulas.  To address this problem, we require the non-ground clauses to be saturated under the Factoring and Resolution inference rule defined below.  These inference rules depend on a selection function, which selects the literals in each clause that may be used in an inference.
A selection function maps a clause to a subset of its literals, just like the trigger function.  A selection function $Sel$ is {\em valid} if, for each clause $C$ and for each $T \subseteq Sel(C)$ with $Vars(T) \not= Vars(C)$, either
$Sel(C) - T$ contains all maximal literals in $C - T$, or
$Sel(C) - T$ contains a negative literal.

Before we give an intuition of this definition, let us give some properties:

\begin{proposition}
For any clause $C$ and
valid selection function $Sel$,
$Vars(Sel(C)) = Vars(C)$.
\end{proposition}
\begin{proof}
We prove the contrapositive.
Suppose $Vars(Sel(C)) \subset Vars(C)$ (i.e., proper).  Let $T = Sel(C)$.  Then $Vars(T) \not= Vars(C)$.  But $Sel(C) - T$ is empty so it cannot contain all maximal literals in $C - T$ and it cannot contain a negative literal. By the definition of a valid selection function, this is a contradiction.
\end{proof}

\begin{proposition}
Let $Sel$ be a selection function such that for all $C$, $Vars(Sel(C)) = Vars(C)$ and either
    (1) $Sel(C)$ contains only negative literals or
    (2) $Sel(C)$ is a singleton set containing the maximum literal in $C$.
Then $Sel$ is a valid selection function.
\end{proposition}

\begin{proof}
    Suppose $Sel(C)$ contains only negative literals.  Let $T$ be a subset of $Sel(C)$ with $Vars(T) \not= Vars(C)$.  Then $T \subset Sel(C)$, so $Sel(C) - T$ contains a negative literal.
    Similarly if $Sel(C)$ contains a single maximum literal in $C$ then $D = \emptyset$ so $Sel(C) - T = Sel(C)$, containing all maximal literals in $C$.
\end{proof}

Selection functions normally select
all maximal literals or a negative literal.
We additionally require $Sel(C)$ to contain all the variables in $C$.
We also require that if some of the literals from the selected set are removed from the clause, without covering all the variables, the remaining selected set must contain all maximal literals in the remaining clause or a negative literal.  In the completeness proof, we will filter our clauses by the ground model, and must ensure that the filtered clauses still have a valid selection function.

In Example~\ref{goodsel} with trigger function $Trig_1$, if the selection function is the same as the trigger function, it is easy to construct an ordering where the selected literal is the largest in $C_1$ and $C_2$.  So this is a valid selection function.
If the trigger function is $Trig_2$, the same is true for $C_1$.  For $C_2$, $Trig_2(C_2)$
contains only negative literals, so the selection function is valid. 

Let us look at one more example.  
Consider the following set of clauses, with an ordering where $r(t_1) > q(t_2) > p(t_3)$ for all terms $t_1, t_2, t_3$.

\begin{example}
\label{countersel}
$$ C_1:  \neg p(X_1) \vee \neg q(X_1) $$
$$ C_2:  p(X_2) \vee \neg q(X_2)  $$
$$ C_3:  \neg p(X_3) \vee q(X_3) $$
$$ C_4:   p(X_4) \vee q(X_4) \vee \neg r(Y_4) $$
\end{example}

Let $Sel_1$ be the selection function such that $Sel_1(C_1) = \{\neg p(X_1)\}$, $Sel_1(C_2) = \{\neg q(X_2)\}$, $Sel_1(C_3) = \{q(X_3)\}$, and $Sel_1(C_4) = \{\neg r(Y_4)\}$.  This selection function cannot be valid, because $Sel_1(C_4)$ does not contain all the variables of $C_4$.  So let $Sel_2$ be a selection function identical to $Sel_1$ on the first three clauses, but $Sel_2(C_4) = \{p(X_4),\neg r(Y_4)\}$.  This selection function is also not valid because in clause $C_4$ if we let $D = \{\neg r(Y_4)\}$, then $Sel_2(C_4) - D = \{p(X_4)\}$, and $p(X_4)$ is neither maximal nor negative in $p(X_4) \vee q(X_4)$.  Finally we define $Sel_3$ to be the same as $Sel_1$ on the first three clauses but $Sel_3(C_4) = \{q(X_4),\neg r(Y_4)\}$.  This selection function is valid.

Given a valid selection function, our inference system will consist of three inference rules.  We defined Instantiation above.  Below we define Resolution and Factoring:

\textbf{Resolution:} \[\frac{A \lor \Gamma \quad \neg B \lor \Delta}{(\Gamma \lor \Delta)\sigma}\] 

where
    (1) $A \lor \Gamma \in ng(S)$ ,
    (2) $\neg B \lor \Delta \in ng(S)$, 
     (3) $A$ is selected  in $A \vee \Gamma$, 
     (4) $\neg B$ is selected in $\neg B \vee \Delta$, and
     (5) $\sigma = mgu(A,B)$.

\textbf{Factoring:} \[\frac{A \lor B \lor \Gamma}{(A \lor \Gamma)\sigma}\] 

where
    (1) $A \vee B \vee \Gamma \in ng(S)$ , 
     (2) $A$ is selected  in $A \vee B \vee \Gamma$, and 
     (3) $\sigma = mgu(A,B)$.

When applying Resolution and Factoring, it is important to remove redundant clauses.  In particular, implementations remove subsumed clauses and tautologies.

\begin{definition}
A clause $C$ {\em subsumes} a clause $D$ if there is a substitution $\sigma$ such that $C\sigma \subseteq D$.  A clause $C$ is a {\em tautology} if there is an atom $A$ such that $A \in C$ and $\neg A \in C$.  A set of clauses $S$ (possibly infinite) is {\em saturated} by Resolution and Factoring if the conclusion of every Resolution and Factoring inference in $S$ is either a tautology or is subsumed in $S$. $S$ is {\em completely saturated} if $S$ is saturated by Instantiation and saturated by Resolution and Factoring.
\end{definition}




In Example~\ref{goodsel}, with selection and trigger function $Trig_1$, the set $\{C_1,C_2,C_3,C_4,C_5\}$ is completely saturated.  If the selection and trigger function are $Trig_2$, the set $\{C_1,C_2,C_3\}$ is saturated by Instantiation but not saturated by Resolution and Factoring.  The result of a Resolution between $C_1$ and $C_2$ is $C_7 = \neg p(X_7,Y_7) \vee p(f(X_7), f(Y_7))$.  We extend the selection and trigger function for this new clause.  Suppose we extend $Trig_2$ so that $Trig_2(C_7) = \{p(f(X_7),f(Y_7))\}$.  Then an instantiation will give us $C_8 = \neg p(a,b) \vee p(f(a), f(b))$.  Extending the model to $\{p(a,b), p(f(a), f(b))\}$ allows us to see that $\{C_1,C_2,C_3,C_7,C_8\}$ is completely saturated.

In Example~\ref{countersel}, suppose we add a ground clause $C_5 = r(a)$, then $\{C_1,C_2,C_3,C_4,C_5\}$ is completely saturated under selection function $Sel_2$, because every Resolution inference yields a tautology, even though the set is unsatisfiable. 
So if we had only required valid selection functions to select a negative literal or all maximal literals in each clause, 
we could not prove completeness, even if we additionally required that the selected literals contain all the variables in the clause.  Using selection function $Sel_3$, the set of clauses is not saturated under Resolution and Factoring.

In the appendices we give axioms of set theory and subsumption,
saturated by Resolution and Factoring, since every Resolution inference is a tautology.  
We gave these examples to z3, with a ground clause.
The default mbqi did not halt, and disabling mbqi returned "unknown", while our method returns "sat".




\section{Completeness Proof}
In this section we will prove the completeness of our inference system. 
Given a set of clauses $S$,
the first step in our completeness proof is to filter the ground instances of $ng(S)$ with a model $M_{g(S)}$ of $g(S)$. 
Let $\N = Filter(Gr(ng(S)),M_{g(S)})$.

We explain $\N$ with an example, where we write $f^n(a)$ to abbreviate $f$ applied $n$ times to $a$.  Note that $S$ is not saturated under Instantiation in this example, although in the proof we only construct $\N$ for completely saturated sets.

\begin{example}
    Let $ng(S) = \{ \neg p(X) \vee p(f(X))\}$. Suppose that we have the model 
    $M_{g(S)} = \{\neg p(f(a)), p(f^3(a))$\}
     of $g(S)$.
    Then $Gr(ng(S)) = 
    \{ \neg p(f^n(a)) \vee p(f^{n+1}(a)) \,\,|\,\, n \geq 0\} $.
    So $\N$ is the set of clauses  $\{\neg p(a), p(f^4(a))\} \cup 
    \{ \neg p(f^n(a)) \vee p(f^{n+1}(a)) \,\,|\,\, n \geq 4\} $.    
\end{example}

Instances of subsumed clauses and tautologies in $S$ are also subsumed clauses and tautologies in $\N$, if they exist in $\N$.

\begin{lemma}
Let $D$ be a clause in $S$.  Let $\theta$ be a ground substitution. 
Let $D' = Filter(D\theta,M_{g(S)})$.  Then
    (a) If $D$ is a tautology
    then either $D'$ is not in $\N$ or
    $D'$ is a tautology.
    (2) If $D$ is subsumed in $S$ then either 
    $D'$ is not in $\N$ or there is a clause $C'$ in $\N$
    such that $C' \subseteq D'$.
\end{lemma}

\begin{proof}
    \begin{enumerate}
        \item Let $D = A \vee \neg A \vee \Gamma$.  If $A\theta \in M_{g(S)}$ or $\neg A\theta \in M_{g(S)}$ then $D' \not\in \N$.  If $A$ is undefined in $M_{g(S)}$ then $D'$ is a tautology.  
        \item Let $C$ be a clause in $S$ and $\sigma$ be a substitution such that $C\sigma \subseteq D$.  Then there is a ground substitution $\theta'$ such that $C\theta' \subseteq D\theta$.  If $D' \in \N$ then there is no literal $L$ in $D\theta$ such that $M_{g(S)} \models \bar{L}$.  So there is no literal $L$ in $C\theta'$
        such that $M_{g(S)} \models \bar{L}$.  
        Let $C' = Filter(C\theta,M_{g(S)})$.  Let $L$ be an arbitrary literal in $C'$.  Then $L$ is undefined in $M_{g(S)}$.  So $L$ is in $D'$.  This implies $C' \subseteq D'$.
    \end{enumerate}
\end{proof}

Subsumption and tautology deletion are instances of the concept of redundancy,\footnote{A clause is redundant if implied by smaller clauses.} 
A $C$ may be redundant in $S$ but $Filter(C)$ not redundant in $\N$,
so our filtering technique does not cover redundancy in full.  However, subsumption and tautology deletion are what is mainly used in practice to control saturation. 

\begin{example}
    Let $S$ be a set of clauses such that 
    $ng(S) = \{ p(X) \vee q(X), \neg q(X), \neg r(X) \vee p(X) \}$
    with $M_{g(S)} = \{ r(a) \}$
    and an ordering such that $p(s) < q(t) < r(u)$ for all terms $s,t,u$.  Then $\neg r(X) \vee p(X)$ is implied by smaller clauses $p(X) \vee q(X)$ and $\neg q(X)$.  But when we apply filtering, we get clauses
    $\{ p(a) \vee q(a), \neg q(a), p(a) \}$, and $p(a)$ is not implied by smaller clauses.
\end{example}

To prove completeness, we will 
let $S$ be a completely saturated set of clauses with $Trig = Sel$.  We show that if $\bot$ is not 
in $S$ and $g(S)$ is satisfiable, then a model of $ng(S)$ can be constructed which is  
compatible with the model of $g(S)$.

First we need some definitions. For a set of clauses $S$,
    let $S_{<C} = \{D \in S \,\,|\,\, D < C\}$ be the set of clauses in $S$ that are smaller than $C$.  
    We will create an interpretation from a set of positive literals.  So, given a set of positive literals $T$ and a set of literals $U$, we define $Int(T,U) = T \cup
    \{ \neg A \,\,|\,\, A \not\in T, 
    (U - T) \cap \{A,\neg A\} \not= \emptyset\}$
    In other words, it is the interpretation where all the atoms in $T$ are true, and every atom in $U$ that has not been made true in $T$ is false.   
For each clause $C \in \N$, we will define $P_{< C}$, $M_{< C}$ and $P_C$ co-recursively. 

\begin{definition}
    Let $\N$ be the clause set defined above, and $C$ be a clause in $\N$. 
    \begin{enumerate}
    \item
    Define $P_{<C}$ as the set of positive literals $\bigcup_{D \in {\N}_{<C}} P_D$, where $P_D$ is defined below.  
    \item 
    Define $P^{\N} = \bigcup_{C \in \N} P_C$, the union of all the $P_C$ defined below.
    \item 
    Let $M_{<C} = Int(P_{<C},C \cup \bigcup {\N}_{<C})$, which means that $M_{<C}$ is the interpretation that makes true all the atoms in $P_{<C}$, and makes false all other atoms in clauses of $\N$ that are smaller than or equal to $C$.
    \item 
    Similarly, let $M^{\N} = Int (P^{\N},\bigcup \N)$.
    \end{enumerate}
%

Simultaneously we define
    $P_C = \{A\}$ for atom $A$ if 
    (1) $M_{<C} \not\models C$, 
    (2) $A$ is the largest literal in $C$,
    (3) $A$ is selected in $C$, and    
    (4) $A$ only occurs once in $C$.
    Otherwise $P_C = \emptyset$.  
    
    If $P_C = \{A\}$, we say that $C$ {\em produces} $A$.
\end{definition}

The completeness proof is similar to the standard proof of completeness of Resolution, except we deal with filtered clauses, so lifting is more complex. Also, we use Instantiation when the filtering removes all the selected literals.  
The next lemma follows from the definition of $P_C$.

\begin{lemma}
\label{inherit}
Let $C$ be a clause in $\N$.  Let $L$ be a literal in $C$.  Then
(1) If $L$ is not maximum in $C$ then $M^{\N} \models L$ if and only if $M_{<C} \models L$.
(2) If $C$ produces $L$ then $L \in M^{\N}$.
\end{lemma}

For the proof below, we will assume that for every $C \in S$ and $L \in C$, where $C' = Filter(C\theta,M_{g(S)})$, then $L\theta$ is selected in $C\theta$ if and only if $L\theta \in C'$ and
$L$ is selected in $C$. 

\begin{theorem}
Let $Sel$ be a valid selection function with $Trig = Sel$.
Suppose that $S$ is completely saturated and $\bot \not\in S$ 
and $g(S)$ is satisfiable.  Let $M_{g(S)}$ be a model of $g(S)$ such that $S$ is saturated by $M_{g(S)}$-Instantiation.
Then $M^{\N} \models \N$ where $\N = Filter(Gr(ng(S)),M_{g(S)})$.
\end{theorem}



\begin{proof}
Suppose $M^{\N}$ is not a model of $\N$. Then let $C$ be the smallest clause in $\N$ such that $M^{\N} \not\models C$.  By the above lemma $M_{<C} \not\models C$.  Since $M^{\N} \not\models C$, then $P_C = \emptyset$.  
So let's examine the reasons why $C$ didn't produce anything.  It must be because of one of the following reasons:

\begin{enumerate}
\item There is a selected literal $\neg A$ in $C$ such that $M_{<C} \models A$.
\item Some literal is selected in $C$, but no negative literals are selected.  Then the largest literal $A$ in $C$ is positive and selected, so $A$ occurs twice in $C$, since $A$ was not produced by $C$.
\item There are no selected literals in $C$.
\end{enumerate}

\textbf{Case 1:} There is a selected literal $\neg A$ in $C$ such that $M_{<C} \models A$.

So $C$ is of the form $C=\neg A \vee \Gamma$. Since $C \in \N$,
$C$ must be undefined in $M_{g(S)}$, and 
there must be a clause $C' = \neg A_1 \vee \Gamma_1 \vee \Gamma_2$ in $S$ and a substitution $\theta$ such that $A = A_1 \theta$, $\Gamma = \Gamma_1 \theta$ and $M_{g(S)} \models \neg \Gamma_2 \theta$.\footnote{We assume variables of different (instances of) clauses are disjoint, so the same substitution can be applied to all clause instances}

There must be a clause $D$ in $\N$ which produced $A$.
So $D$ is of the form $A \vee \Delta$, which is undefined in $M_{g(S)}$, and a clause $D' = A_2 \vee \Delta_1 \vee \Delta_2$ in $S$ such that $A = A_2\theta$,
$\Delta=\Delta_1 \theta$, and $M_{g(S)} \models \neg \Delta_2 \theta$. 

Since $A_1\theta = A_2\theta = A$, the following Resolution inference exists in $S$:

 \[\frac{\neg A_1 \lor \Gamma_1 \lor \Gamma_2 \quad A_2 \lor \Delta_1 \lor \Delta_2}{(\Gamma_1 \lor \Gamma_2 \lor \Delta_1 \lor \Delta_2)\sigma}\] 

Since $\sigma = mgu(A_1,A_2)$, $\Gamma_1 \theta \lor \Gamma_2 \theta \lor \Delta_1 \theta \lor \Delta_2 \theta$ is an instance of 
the conclusion of this inference.  Filtering this clause with $M_{g(S)}$ gives us $\Gamma \vee \Delta$.  Since $\Delta < A < \neg A$, this clause is smaller than $C$.  When $D$ produced $A$, it must have been because $M_{<C} \not\models \Delta$, so $M^{\N} \not\models \Delta$.  Therefore $M^{\N} \not\models \Gamma \vee \Delta$, so it cannot be a tautology.  Either $\Gamma \vee \Delta$ is in $\N$ or is subsumed in $\N$ by Lemma~\ref{inherit}.  In both cases, we get a smaller counterexample, a contradiction. 
 
\textbf{Case 2:} The largest literal $A$ is positive and selected, and occurs twice in $C$.

 Therefore, $C$ is of the form $C = A \vee A \vee \Gamma$, and there must be a clause $A_1 \vee A_2 \vee \Gamma_1 \vee \Gamma_2$ in $S$
where $A_1 \theta = A_2 \theta = A$, $\Gamma_1 \theta = \Gamma$, and $M_{g(S)} \models \neg \Gamma_2 \theta$. Since $A_1\theta = A_2\theta$, the following Factoring inference exists in $S$:

 \[\frac{A_1 \lor A_2 \lor \Gamma_1 \lor \Gamma_2}{(A_1 \vee \Gamma_1 \vee \Gamma_2) \sigma}\] 

Since $\sigma = mgu(A_1,A_2)$,
$A\theta \lor \Gamma_1 \theta \lor \Gamma_2 \theta$ is an instance of the conclusion of this inference.  Filtering this with $M_{g(S)}$
gives us $A \lor \Gamma$, which cannot be a tautology.
So it is in $\N$ or is subsumed in $\N$.  In both cases, we get a smaller counterexample, a contradiction. 


\textbf{Case 3:} There are no selected literals in $C$.

There must be a clause $C' \vee \Gamma$ in $S$, where $C'\theta = C$,
$C$ is undefined in $M_{g(S)}$, $M_{g(S)} \models \neg 
\Gamma\theta$, and all of the selected literals of this clause are 
in $\Gamma$.  
Let $\Gamma = \Delta \vee L_1 \vee \cdots \vee L_n$ 
where $L_1 \cdots L_n$ are
the selected literals in $\Gamma$.  
Then there exist literals 
${L_1}', \cdots, {L_n}'$ in $M_{g(S)}$ such that $\bar{L_i}\theta =
{L_i}'$ for all $1 \leq i \leq n$.  So the following Instantiation
inference exists in $S$:

 \[\frac{L_1 \vee \cdots L_n \vee \Delta \vee C'}
        {(L_1 \vee \cdots L_n \vee \Delta \vee C') \theta}\] 

The conclusion of an Instantiation inference is a ground clause. 
So $M_{g(S)} \models L_1\theta \vee \cdots \vee L_n\theta \vee \Delta\theta \vee C'\theta$.
But
$M_{g(S)} \models \neg \Gamma\theta$ and $C$ is undefined in $M_{g(S)}$, 
so this is a contradiction.
\end{proof}

We can combine $M^{\N}$ with $M_{g(S)}$ to get a model of $S$.

\begin{corollary}
\label{completeness}
Let $Sel$ be a valid selection function with $Trig = Sel$.
Suppose that $S$ is completely saturated and $\bot \not\in S$
and $g(S)$ is satisfiable.  Let $M_{g(S)}$ be a model of $g(S)$ such that $S$ is saturated by $M_{g(S)}$-Instantiation.
Let $\N = Filter(Gr(ng(S)),M_{g(S)})$.  Then $M^{\N}$ is compatible with $M_{g(S)}$ and $M^{\N} \cup M_{g(S)} \models S$ 
\end{corollary}

\begin{proof}
    $M^{\N}$ and $M_{g(S)}$ are compatible because all the literals in $\N$ are undefined in $M_{g(S)}$, and $M^{\N}$ only contains literals in $\N$.  So $M^{\N} \cup M_{g(S)}$ is consistent, and is an interpretation. 
    
    Let $C$ be a clause in $S$ and let $\theta$ be a ground substitution.  Let $C' = Filter(C\theta,M_{g(S)})$.  If $C'$ is in $\N$ then $M^{\N} \models C'$ and therefore $M^{\N} \models C\theta$.  If $C'$ is not in $\N$ then $M_{g(S)} \models C\theta$.
    This means that $M^{\N} \cup M_{g(S)} \models ng(S)$.  Since $M_{g(S)} \models g(S)$, then $M^{\N} \cup M_{g(S)} \models S$.
\end{proof}

\section{Decision Procedure and Complexity Results} 
Next we assume the clauses are saturated by Resolution and Factoring, 
and find cases where they can be further finitely saturated where Instantiation only produces finitely many new ground clauses. 
Then SAT Solving plus Instantiation is a decision procedure for the theory of $ng(S)$, because the SAT solver will also produce a model in finite time.  If there are only polynomial many ground clauses produced by Instantiation, and if we have a class of problems where the SAT solver always produces a model in polynomial time, then saturation by Instantiation can be done in polynomial time.  Since we need to look at the SAT solver more closely, we will give an inference system to describe the SAT solver along with the Instantiation rule.

We will model a CDCL SAT solver\cite{DBLP:series/faia/0001LM21}  
as an inference system, with each state represented by 
a triple $<G,M,LC>$, where 
\begin{enumerate}
\item
$G$ is the current set of ground clauses.  This set will be expanded as new clauses are learned.
\item
$M$ is a list of literals, representing the partial interpretation created by the SAT Solver.  If an atom $A$ is in the list, it indicates that $A$ has been set to true.  If $\neg A$ appears in the list, it indicates that $A$ is currently set to false.  The literals in the list are in reverse order of how they were set, i.e., the first literal in the list is the last one set.
\item 
$LC$ is either a singleton set of one clause or the empty clause.
When a conflict is discovered, then $LC$ will contain the conflict clause, which will be continually modified until it becomes the learned clause.  If $LC = \emptyset$ then there is no current conflict.
\end{enumerate}

In our inference system, we will model the rules Decide, Unit Propagate, Backjump and Clause Learning.  We are not modelling Forget and Restart.  The Decide rule sets the value of a literal if no other rule applies.  We assume it will always set a literal to false, as many SAT solvers do. 


If $M$ is a list of literals $[L_n,\cdots,L_1]$, we write $M[L_i\cdots]$ to refer to $[L_i,\cdots,L_1]$.
Given $M$, we define a function $Count_M$ from literals to positive integers so that if $L \in M$ then $Count_M(L) = |M[L\cdots]|$ otherwise $Count_M(L) = \omega$.  
We define an ordering $\leq_M$ so that $L \leq_M L'$ if $Count_M(L) \leq Count_M(L')$.  In other words, $L \leq_M L'$ if the truth value of $L$ was determined before the truth value of $L'$ or if $L'$ is undefined. 
Let $<_M$ be the strict version of $\leq_M$.

We define a function $Sort_M$ that maps each clause $C$ to a permutation of $C$ that is sorted in descending order according to $\leq_M$.  For a set of clauses $G$, $Sort_M(G) = \{Sort_M(C) \,\,|\,\, C \in G\}$.  


We define a function called $Level_M$, mapping literals to non-negative integers so that if $L \in M$ then $Level_M(L)$ is the number of literals in $M[L\cdots]$ that were added to $M$ by the Decide rule.  If $L \not\in M$ then $Level_M(L) = \omega$.


Lists are represented with square brackets. The empty list is $[]$.  We use the colon to add an element to a list. 


The initial state of our inference system is $<g(S),
[],\emptyset>$.  We write inference rules in the form 
$<G,M,LC> \Rightarrow <G',M',LC'>$, to represent the 
fact that the SAT solver can move from the first state 
to the second state.  The SAT solver is don't care 
nondeterministic, in the sense that there is no need 
for backtracking.
The inference rules are as follows:

\textbf{Decide:}
$<G,M,\emptyset>$ 
$\Rightarrow$ 
$<G,\neg A : M, \emptyset>$
\hspace{1cm}
where 
\begin{enumerate}
\item 
There is no clause $L \vee \Gamma$ in $Sort_M(G)$ such that $M \models \neg \Gamma$.
\item
$A \not\in M$ and $\neg A \not\in M$.
\end{enumerate}

The conditions enforce that Decide is only applied when none of the rules below are applicable, except for possibly Instantiate.  Note that a negative literal is always set true.  The definition of $Level_M$ ensures that the Decide rule increases the level of literals. 

\textbf{Propagate:}
$<G,M,\emptyset>$ 
$\Rightarrow$ 
$<G,L : M, \emptyset>$
\hspace{1cm}
where 
\begin{enumerate}
\item 
There is a clause $L \vee \Gamma$ in $Sort_M(G)$ such that $M \models \neg \Gamma$.  We say that $L \vee \Gamma$ produced $L$ in $M$.
\item
$L \not\in M$ and $\bar{L} \not\in M$.
\end{enumerate}

\textbf{Conflict:}
$<G,M,\emptyset>$ 
$\Rightarrow$ 
$<G,M, \{C\}>$
\hspace{1cm}
where 
\begin{enumerate}
\item 
$C \in G$.
\item 
$M \models \neg C$.
\end{enumerate}

\textbf{Backjump:}
$<G,M,\{C\}>$ 
$\Rightarrow$ 
$<G,M, \{\Delta \vee \Gamma\}>$
\hspace{1cm}
where 
\begin{enumerate}
\item 
$Sort_M(C) = \bar{L} \vee \Gamma$.
\item 
There exists $L' \in \Gamma$ such that $Level(L') = Level(L)$.
\item 
$Clause(L) = L \vee \Delta$.
\end{enumerate}

The Backjump rule will be applied repeatedly until 
the literal at the maximal level of the clause is the only literal in the clause at that level.
 Then the clause will be learned.  Backjump and Learn are often defined in terms of an implication graph, but it can also be expressed using Resolution, as we have here.

\textbf{Learn:}
$<G,M,\{C\}>$ 
$\Rightarrow$ 
$<G \cup \{C\},M', \emptyset>$
\hspace{1cm}
where 
\begin{enumerate}
\item 
Backjump does not apply.
\item 
$C \not= \bot$.
\item 
If $C$ is a unit clause then $M' = []$. 
\item
If $Sort_M(C) = L' \vee L \vee  \Gamma$
then $M' = M[\bar{L}\cdots]$.
\end{enumerate}

When we learn the clause, we go back to the place where the second largest literal in the clause was just set to false.

\textbf{Instantiate:}
$<G,M,\emptyset>$ 
$\Rightarrow$ 
$<G \cup \{C\}, M', \emptyset>$
\hspace{1cm}
where 
\begin{enumerate}
\item 
$C$ is the conclusion of an $M$-instantiate inference and $C \not\in G$.

\item 
If $C$ is a unit clause then $M' = []$.
\item
If $Sort_M(C) = L' \vee L \vee  \Gamma$
and $M \models \neg (L \vee\Gamma)$
then $M' = M[\bar{L}\cdots]$. 
\item 
If $Sort_M(C) = L' \vee L \vee  \Gamma$
and $M \not\models \neg (L \vee \Gamma)$
then $M' = M$. 
\end{enumerate}

The Instantiate rule can be applied at any time.  But it must be applied if all the atoms have been given a truth value.  Just like the Learn rule, it backs up to just after the second largest literal has been set to false.  If $M$ does not imply $\neg \Gamma$, there is no need to back up.

\textbf{Succeed:}
$<G,M,\emptyset>$ 
$\Rightarrow$ 
(SAT,M)
\hspace{1cm}
where 
\begin{enumerate}
\item 
There is no clause $C$ in $G$ such that $M \models \neg C$.
\item 
All atoms in $G$ are defined by $M$.
\item 
The Instantiate rule does not apply.
\end{enumerate}
                                                    
\textbf{Fail:}
$<G,M,\bot>$ 
$\Rightarrow$ 
UNSAT

The inference rules terminate when we have determined SAT or UNSAT.  For SAT, we also return a model.

Given a set of clauses $S$, a sequence of inference steps starting with $(g(S),[],\emptyset)$
is called a {\em derivation from $S$}.
If the derivation ends with $(SAT,M)$ or $UNSAT$, it is called a {\em terminating derivation} from $S$.

The following theorem for non-ground clauses is well-known.  It only says that SAT solvers always halt and give the correct answer.

\begin{theorem}
    Let $S$ be a set of ground clauses. Then every sequence of inference steps from $(g(S),[],\emptyset)$is finite.  
    If $S$ is unsatisfiable then every derivation from $S$ terminates with $UNSAT$.  If $S$ is satisfiable then every derivation  from $S$ terminates with $(SAT,M)$, where $M$ is a model of $S$.
\end{theorem}

If $S$ is a set of clauses, and the
non-ground clauses are saturated by Resolution and Factoring, then Instantiation is a decision procedure for $ng(S)$ under certain conditions.  
For example, if the selection function selects a single maximum literal in each clause.

\begin{theorem}
Let $Sel$ be a valid selection function such that a single maximum literal is selected in each clause of $ng(S)$, with $Trig = Sel$.  
    Let $S$ be a set of clauses saturated by Resolution and Factoring.  Then every sequence of inference steps from $(g(S),[],\emptyset)$ is finite.
        If $S$ is unsatisfiable then every derivation from $S$ terminates with $UNSAT$.  If $S$ is satisfiable then every derivation from $S$ terminates with $(SAT,M)$, where $M$ is a model of $S$.
\end{theorem}

\begin{proof}
First we want to prove that the every sequence of inference steps is finite.  We do this by constructing a forest of trees based on the Instantiation rule inferences.  We create a node 
 labelled by each ground atom that appears in the saturation by Instantiation.  
    Since the maximum literal $A$ is selected in each non-ground clause, if an instantiation creates a new ground atom $B\theta$ in $C\theta$ then $B\theta < A\theta$.
    We create an edge from the node labelled $A\theta$ to the node labelled $B\theta$.

   The forest of trees is finite because:
    \begin{enumerate}
        \item There are finitely many trees, since the root of each tree is labelled with an initial ground atom.
        \item If a node is labelled with $A\theta$, then any edge out of this node is created by a clause $C \in ng(S)$ such that $A$ or $\neg A$ is in $C$.  The node at the other end must be labelled with $B\theta$, where $B$ is in $C$.  There are only finitely many possibilities, so each node has a finite number of children. 
        \item Each branch is of finite length.  This is because the ordering is well founded.  
    \end{enumerate}
Since the forest of trees is finite, only finite many Instantiations are performed, and therefore the procedure will halt.

We need to show that the procedure produces the correct answer when it halts.  By soundness, the UNSAT case gives a correct answer.  For the SAT case note that we have constructed a model of all the grounds clauses.  So the clauses are completely saturated.  By Corollary~\ref{completeness}, $S$ is satisfiable, and $M$ is a model of $S$.

\end{proof}

We also get a decision procedure if 
the ordering used is order isomorphic to $\omega$.

\begin{theorem}
Let $Sel$ be a valid selection function such that all maximal literals are selected in each clause, with an ordering that is order isomorphic to $\omega$.
Suppose that $Trig = Sel$ and that 
$S$ is saturated by Resolution and Factoring.
  Then every sequence of inference steps from $(g(S),[],\emptyset)$ is finite.
        If $S$ is unsatisfiable then every derivation from $S$ terminates with $UNSAT$.  If $S$ is satisfiable then every derivation from $S$ terminates with $(SAT,M)$, where $M$ is a model of $S$.
\end{theorem}

\begin{proof}
This proof is a little simpler than the previous one.  We simply observe that all new atoms created are smaller than an initial atom.  So only finitely many atoms can be created, and then there are only finitely many Instantiations, which means the procedure halts.  When the procedure halts, $S$ is completely saturated, so the SAT and UNSAT results are correct.
\end{proof}

If the selection function selects a negative literal, then the Instantiation rule may not halt.

\begin{example}
    Consider Example~\ref{goodsel} where $Trig(C_1) = \{\neg p(X_1,Y_1)\}$ and $Trig(C_2) = \{\neg q(X_2,Y_2)\}$.  This is saturated by Resolution and Factoring, but Instantiation with ground clause $p(a,a)$ creates infinitely many clauses.
\end{example}

Even if all maximal literals are selected in each clause, there may still be infinitely many instantiations, as in the following theory, with an ordering such that $p(s) > q(t)$ for all $s$ and $t$.

\begin{example}
$$C: \neg p(X) \vee \neg q(Y) \vee q(f(Y))$$  
\end{example}

Suppose that $Sel(C) = \{\neg p(X), \neg q(Y)\}$.  This is a valid selection function, and $p(X)$ is the maximum literal.  But suppose we have $p(c)$ and $q(c)$ in the ground model.  Instantiation will create $q(f^n(c))$ for all $n$, all literals smaller than $p(c)$.

We would also like to determine conditions where 
Instantiation halts in polynomial time, given that $S$ is already saturated by Resolution and Factoring.
This requires that only polynomially many instantiations are computed, and that the SAT solver runs in polynomial time.  

If $ng(S)$ is saturated by Resolution and Factoring, and a derivation has only polynomimally many Instantiation inferences and polynomially many Learn inferences then that derivation is computable in polynomial time.

\begin{lemma}
\label{poly}
    Let $Sel$ be a valid selection function 
and $Trig = Sel$.  Suppose 
$ng(S)$ is saturated by Resolution and Factoring.  If a derivation from $S$ has only polynomially many Instantiation steps and polynomially many Learn steps
then that derivation can be computed in polynomial time. 
\end{lemma}

\begin{proof}
    First we show that there are only polynomially many steps between Instantiation and Learn inferences.  We can see that between Instantiation and Learn inferences we must have a sequence of Decide and Propagate steps followed by a sequence of Backjump steps.  That sequence has at most $n$ Decide and Propagate steps, where $n$ is the number of atoms, since each Decide and Propagate step makes the model larger.  There are also at most $n$ Backjump steps, because Backjump removes the largest literal from the clause that will eventually be learned.  

    Since we assumed that the number of Instantiation and Learn inferences is polynomial, there are only polynomially many steps in the derivation.  Since each individual step can be done in polynomial time, the whole derivation can be computed in polynomial time.
\end{proof}

Next we prove some relatively simple properties of CDCL SAT sovlers.
If a clause $C$ is used to produce a literal in the Unit Propagation rule then the two largest literals in $C$ will have the same level. 

\begin{lemma}
\label{prop}
If $L' \vee L \vee \Gamma$ produced $L'$ in $M$ then 
 $Level_M(\bar{L}) = Level_M(L')$.
\end{lemma}

\begin{proof}
We know that $M[\bar{L}\cdots] \models \neg (L \vee \Gamma)$.  Then $L'$ will be produced sometime after $\bar{L}$ is produced, but
before a Decide inference is performed.  Therefore $Level_M(\bar{L}) = Level_M(L')$.
\end{proof}

Next we prove that conflict clauses also have their two largest literals at the same level.

\begin{lemma}
\label{conflict}
If a clause $L' \vee L \vee \Gamma$ is a conflict clause, then 
$Level_M(\bar{L}) = Level_M(\bar{L'})$.
\end{lemma}

\begin{proof}
As in the previous proof, 
if $Level_M(\bar{L}) \not= Level_M(\bar{L'})$ then $L'$ would have been propagated at an earlier level.  
\end{proof}

Conflicts at level 0 must be either unit clauses or the empty clause.

\begin{lemma}
\label{levelzero}
If there is a conflict at level 0, the learned clause is either a unit clause or an empty clause.
\end{lemma}

\begin{proof}
    A clause with at least two literals is only learned when its two largest literals have a different level.  But in a conflict at level 0, the literals must all be at level 0, since there is no smaller level.
\end{proof}

If clause $C$ is a conflict clause not at level 0 then $C$ has at least two literals.

\begin{lemma}
\label{levelposititive}
    A conflicting clause at level greater than 0 
    has at least two literals.
\end{lemma}

\begin{proof}
    A clause with only one literal would have been learned at level 0.
\end{proof}


For Horn clauses, we rely on our assumption that the initial decision about an atom is to make it false.  It is well-known that satisfiability of ground Horn clauses can be decided in polynomial time, but we are not aware of any results that CDCL SAT solvers solve Horn clauses in polynomial time.  




For Horn clauses, we need one more lemma saying that all conflicts occur at level 0.

\begin{lemma}
\label{hornlemma}
    If all clauses in $S$ are Horn clauses then there are no conflicts at a level greater than 0.
\end{lemma}

\begin{proof}
        Since all non-ground clauses in $S$ are Horn clauses, all instantiated clauses are Horn clauses.  Also, all clauses learned at level 0 have fewer than two literals, so they are Horn clauses.

        The first literal decided at a level greater than 0 must be a negative literal.  By induction, and the fact that all clauses are Horn, all propagated literals must also be negative.  By Lemma ~\ref{conflict}, a conflict clause must be of the form $L' \vee L \vee \Gamma$, where $Level_M(\bar{L'}) = Level_M(\bar{L}) > 0$.  Since all literals propagated at level > 0 are negative, then $L'$ and $L$ are positive.  So the conflict clause is not Horn.  
\end{proof}

\begin{theorem}
Let $<$ be a polynomial ordering.
Let $Sel$ be a valid selection function such that all maximal literals are selected in each clause of $ng(S)$, with $Trig = Sel$.  
If $S$ is saturated by Resolution and Factoring,
and $S$ only contains Horn clauses,  then every derivation from $S$ can be computed in polynomial time, if we consider $ng(S)$ to be fixed.
        If $S$ is unsatisfiable then every derivation from $S$ terminates with $UNSAT$.  If $S$ is satisfiable then every derivation from $S$ terminates with $(SAT,M)$, where $M$ is a model of $S$.

\end{theorem}

\begin{proof}
Because the ordering is a polynomial ordering, only polynomially many new ground atoms can be created.  Therefore only polynomially many ground atoms can exist.  Each Instantiation matches a set of literals from a non-ground clause with ground literals, so there are only a polynomial number of Instantations, since $ng(S)$ is fixed and therefore the number of selected literals in each clause is fixed.

Since all clauses are Horn clauses, Instantiation also produces Horn clauses, so all ground clauses are Horn clauses.  By Lemma~\ref{hornlemma}, the only conflicts that can occur are at level 0.  By Lemma~\ref{levelzero}, the learned clauses must be unit clauses or the empty clause.  There are only a linear number of those, so there are only a linear number of Learn inferences.  By Lemma~\ref{poly}, each derivation can be computed in polynomial time.

        The UNSAT case gives a correct answer by soundness.  For the SAT case note that we have constructed a model of all the ground clauses.  So the clauses are completely saturated.  By Corollary~\ref{completeness}, $S$ is satisfiable, and $M$ is a model of $S$.

\end{proof}

Clause learning is crucial for 2SAT.
Satisfiability of 2SAT 
can be decided in polynomial time, but we have not seen any results that CDCL SAT solvers solve 2SAT in polynomial time.  

\begin{lemma}
\label{2SATlemma}
    If all clauses in $S$ have at most two literals then all learned clauses have fewer than two literals.
    \end{lemma}

\begin{proof}
    Since all non-ground clauses in $S$ have at most two literals, all instantiated clauses have at most two literals. 
    
    By Lemma~\ref{levelzero}, a clause learned at level 0 has zero or one literal.  So we consider a clause $C$ conflicting at a level greater than 0.  By Lemma~\ref{levelposititive}, $C = L' \vee L$, and by Lemma~\ref{conflict}, $Level_M(\bar{L'}) = Level_M(\bar{L})$. A clause used for Backjump must also have all its literals at the same level by Lemma~\ref{prop}, and the fact that all clauses have at most two literals.  So, by induction, a series of Backjumps will yield a clause with all its literals at the same level.  Since a learned clause with two literals must have those literals at different levels, a learned clause cannot have two literals.    
\end{proof}

\begin{theorem}
Let $<$ be a polynomial ordering.
Let $Sel$ be a valid selection function such that all maximal literals are selected in each clause of $ng(S)$, with $Trig = Sel$.  
If $S$ is saturated by Resolution and Factoring,
and $S$ only contains clauses with at most two literals,  then every derivation from $S$ can be computed in polynomial time.
        If $S$ is unsatisfiable then every derivation from $S$ terminates with $UNSAT$.  If $S$ is satisfiable then every derivation from $S$ terminates with $(SAT,M)$, where $M$ is a model of $S$.
\end{theorem}

\begin{proof}
Because the ordering is a polynomial ordering, only polynomially many new ground atoms can be created.  Therefore only polynomially many ground atoms can exist.  Each Instantiation matches a set of literals from a non-ground clause with ground literals, so there are only a polynomial number of Instantations.

Since all clauses have at most two literals, Instantiation also produces clauses with at most two literals, so all ground clauses have at most two literals.  By Lemma~\ref{2SATlemma}, all learned clauses have fewer than two literals.
There are only a linear number of those, so there are only a linear number of Learn inferences.  By Lemma~\ref{poly}, each derivation can be computed in polynomial time.

        The UNSAT case gives a correct answer by soundness.  For the SAT case note that we have constructed a model of all the ground clauses.  So the clauses are completely saturated.  By Corollary~\ref{completeness}, $S$ is satisfiable, and $M$ is a model of $S$.
\end{proof}

Let us look at Example~\ref{goodsel}, with selection and trigger function $Trig_1$, 
with $p(s_1,s_2) \leq p(t_1,t_2)$ if $s_1$ is a subterm of $t_1$ and $s_2$ is a subterm of $t_2$. We do the same for $q$.
This is a set of Horn clauses (and also 2SAT), so CDCL SAT solving plus Instantiation will solve this theory in polynomial time.  

\section{Related Work}
The paper \cite{DBLP:journals/jar/DrossCKP16}
discusses how a good selection of triggers will give a decision procedure.  Their approach is somewhat different from ours.  The user needs to supply a correctness and termination proof that the trigger choice will give a decision procedure.  Our method is automatic, and inherits the trigger selection function directly from the selection function used in saturation.
Good trigger selection is discussed from a practical point of view in 
\cite{DBLP:conf/cav/LeinoP16,10.1145/1670412.1670416}.  

Other papers suggest other approaches to quantifiers instead of triggers.  Some successful approaches are Model-Based Quantifier Instantiation \cite{DBLP:conf/cav/GeM09} for satisfiable problems, and Conflicting Instances\cite{DBLP:conf/tacas/BarbosaFR17,DBLP:journals/tplp/ReynoldsTB17} for unsatisfiable problems.
Several other approaches have been proposed and implemented \cite{DBLP:conf/lpar/Rummer12,DBLP:conf/tacas/ReynoldsBF18,DBLP:conf/frocos/FontaineS21,DBLP:conf/fmcad/ReynoldsTM14,DBLP:conf/tacas/NiemetzPRBT21,DBLP:conf/vmcai/HoenickeS21}. 
Our paper only  deals with first-order theories without equality,
whereas the above mentioned papers consider other SMT theories.

Other papers have used Saturation under Ordered Resolution \cite{DBLP:journals/jacm/BasinG01},
as a way to show that a first-order Theory is a Local Theory\cite{DBLP:journals/tocl/GivanM02}
meaning that the only instantiation necessary are to replace variables with terms smaller than initial ground terms.  In this approach, all possible instantiations are made at the beginning.
This approach was further extended in \cite{DBLP:conf/cade/Sofronie-Stokkermans05} to cover other theories, and was extended further in several papers, e.g. \cite{DBLP:conf/tacas/IhlemannJS08}.  But these extensions still require instantiating all the instances at the beginning.  Finally, in \cite{DBLP:conf/cav/Jacobs09}, an approach was implemented where 
instantiations are only made when necessary.  But that approach is based on the instance generation method of \cite{DBLP:conf/lics/GanzingerK03,DBLP:conf/cade/Korovin08}, which is not the same as the SMT method.  Finally, in \cite{DBLP:conf/cav/BansalR0BW15}, the local theory method was implemented in an SMT setting.  These ideas don't involve triggers.

Another related technique is for an SMT solver to call a first-order theorem prover \cite{DBLP:journals/jar/LynchTT13,DBLP:conf/cade/MouraB08a,DBLP:conf/cav/Voronkov14}.  Our method is different in that we do not need a first-order theorem prover after saturation of the first-order clauses.


\section{Conclusion}
We analyzed the completeness of the trigger selection function for SMT solvers with only a first-order theory and no other theories. If the first-order theory is saturated by Resolution and Factoring, with a valid selection function identical to the trigger function, then if Saturation by Instantiation gives a model of the ground clauses, that is also a model of those clauses modulo the first-order theory.  Saturation by Instantiation is guaranteed to halt if the Selection function selects a single maximum literal in each clause, or if all maximal clauses are selected using an ordering isomorphic to $\omega$.  If it is also a polynomial ordering, then Saturation by Instantiation is guaranteed to halt in polynomial time if all clauses are Horn Clauses, or if all clauses contain at most two literals.

We began this research by noticing that SMT solvers return "unknown" on problems that seem to be easily shown to be satisfiable.  We hope that implementers of SMT solvers will use our results to return "satisfiable" in more cases.  It requires no change to the SMT process.  The only change is in the pre-processing, where the SMT solver checks if the first-order classes are saturated by a valid selection function, and uses the identical trigger function.  At the end, if no contradiction is found, the SAT solver will return "satisfiable", and also return a model modulo the first-order theory.

We have implemented an SMT solver that, given a satisfiable saturated first-order theory, will detect satisfiability and return a ground model.  We experimented with our SMT solver using some first-order theories presented in the appendices.  Since this is a new SMT solver, we don't expect it to be competitive in speed with existing SMT solvers.  However, this paper is not about increasing the speed of an SMT solver.  It is about making SMT solvers more precise.

We plan lots of future work on this subject.  To make this useful, we need to extend the results to more theories.  We are working on extending it to equality with uninterpreted function symbols.  
Possibly, previous research from  \cite{DBLP:conf/frocos/HorbachS13,DBLP:journals/iandc/LynchRRT11,DBLP:conf/csl/Lynch04,DBLP:journals/jar/BonacinaLM11} 
could help with this.
Later work will be to extend it to other specialized theories.  

Even in the non-equational case, there are many unanswered questions.  For example, can this be extended to theories which cannot be saturated under Resolution.  These results basically give Herbrand models.  There may be ways to use other models to strengthen these results.  
There are several more detailed results that are not answered in this paper.  Does the proof technique work for all cases of redundancy, not just subsumption and tautology deletion?  Other decision procedures may be possible by loosening the restrictions on the ordering.


\bibliographystyle{unsrtnat}
\bibliography{references}  

\appendix

\section{Subsumption Theory}

The theory of Subsumption is given below.  First we define the matching predicate $m$, for a theory with a binary predicate symbol $f$, a unary predicate symbol $g$ and a constant $a$.  We assume there are no other symbols in the terms except for variables.  Instead of including type information, we create a unary function symbol $v$ so that a variable $X$ is represented as $v(x)$.  Then $m(s,t)$ is true if and only if $s$ matches $t$, and the model will give the assignment for the variables.

We also include a predicate $s$ for subsumption.  Clauses are represented by lists of atoms (we don't deal with negation, but that could easily be added).  A list is recursively defined as either "nil" or a pair $p(s,t)$ where $s$ is an atom and $t$ is a list.  In this theory, for subsumption, the second clause must contain at most two atoms, but it is easy to extend it to more atoms.  Then $s(t_1,t_2)$ is true if and only if clause $t_1$ subsumes clause $t_2$.

This set of clauses is saturated by Resolution and Factoring.
The trigger in every clause is the literal with the most symbols.  The only exception is Clause 13, where two literals are selected, because they both have the most symbols.  This determines a polynomial ordering so this gives a decision procedure for this theory.  It is a polynomial time procedure for the theory of matching, since all matching clauses (the first 13 clauses) are Horn.  But it is not a polynomial time procedure for subsumption, since there are non-Horn clauses and there are clauses with more than two literals.  In fact, the theory of subsumption is NP-complete.

Also see \cite{DBLP:conf/fmcad/RathBK22} for another paper on representing subsumption using SAT.  That paper does not use a first-order theory.

\begin{enumerate}
\item
$\neg m( f(X_1, Y_1), f(X_2 , Y_2)) \vee m( X_1,X_2 )$
\item
$\neg m( f(X_1 , Y_1), f(X_2 , Y_2 )) \vee m( Y_1,Y_2 )$
\item
$\neg m( X_1, X_2 ) \vee \neg m( Y_1, Y_2 ) \vee m( f(X_1 , Y_1), f(X_2 , Y_2 ))$

\item
$\neg m( g(X), g(Y) ) \vee m( X, Y )$
\item
$\neg m( X, Y) \vee m( g(X), g(Y) )$
 
\item
$m( a, a )$

\item
$\neg m( f(X,Y), a )$
\item
$\neg m( a, f(X,Y) )$
\item
$\neg m( f(X,Y), g(Z) )$
\item
$\neg m( g(Z), f(X,Y) )$
\item
$\neg m( g(X), a )$
\item
$\neg m( a, g(X) )$

\item
$\neg m( v(X), Y ) \vee \neg m( v(X), Z) \vee m( Y, Z )$

\item
$s(nil,C)$
\item
$\neg s( p(X,Y), nil )$

\item
$ \neg s( p(L,C) , p(K,nil)) \vee m(L,K)$
\item
$ \neg s( p(L,C) , p(K,nil)) \vee s(C,nil) $
\item
$ \neg m(L,K) \vee \neg s(C,nil)  \vee s( p(L,C), p(K,nil)) $

\item
$ \neg s( p(L,C) , p(K1,p(K2,nil)) \vee m(L,K1) \vee m(L,K2) $
\item
$ \neg s( p(L,C) , p(K1,p(K2,nil)) \vee m(L,K1) \vee s(C,p(K1,nil)) $
\item 
$ \neg s( p(L,C) , p(K1,p(K2,nil)) \vee s(C,p(K2,nil)) \vee m(L,K2) $
\item
$ \neg s( p(L,C) , p(K1,p(K2,nil)) \vee s(C,p(K2,nil)) \vee   s(C,p(K1,nil)) $

\item
$ \neg m(L,K1) \vee \neg s(C,p(K2,nil)) \vee s( p(L,C) ) , p(K1,p(K2,nil)) $
\item
$ \neg m(L,K2) \vee \neg s(C,p(K1,nil)) \vee s( p(L,C) ) , p(K1,p(K2,nil)) $

\end{enumerate}

\section{Set Theory}

Below we present a subset of Set Theory without Subset.
We define the "union", "intersect" and "complement" function.  $mem(s,t)$ is true if and only if $s$ is a member of $t$.  We did not include the "subset" function, because we were unable to saturate that theory.

Inspired by \cite{DBLP:journals/cacm/HeuleK17}, we extended the theory to solve the triple sum problem.  The problem we considered is a set of triples, where each element is an integer from $1$ to $n$.  We have three sets: $a$, $b$, and $c$.  We want to distribute each triple $(m,n,p)$ into these sets so that $m$, $n$ and $p$ are not all in the same set.  We write $number(n)$ to indicate that $n$ is a number.  We write $triple(m,n,p)$ to indicate that $(m,n,p)$ is a triple.  $distinct(m,n)$ means that $m$ and $n$ are in different sets.  $both(m,n,s)$ means that $m$ and $n$ are both in set $s$.

This set of clauses is saturated by Resolution and Factoring.  There is a way to select a literal in each clause with an ordering that is order isomorphic to $\omega$.  But some clauses are not Horn, and some clauses have more than two literals.  In \cite{DBLP:journals/cacm/HeuleK17}, the problem for two sets is represented as a SAT problem without a first-order theory.  It could easily be extended to cover three sets.

\begin{enumerate}
\item
$ \neg mem(E,X) \vee mem(E,union(X,Y)) $
\item
$ \neg mem(E,Y) \vee mem(E,union(X,Y)) $
\item
$ \neg mem(E,union(X,Y)) \vee mem(E,X) \vee mem(E,Y) $
\item
$ \neg mem(E,int(X,Y)) \vee mem(E,X) $
\item
$ \neg mem(E,int(X,Y)) \vee mem(E,Y) $
\item
$ \neg mem(E,X) \vee \neg mem(E,Y) \vee mem(E,int(X,Y)) $
\item
$ mem(E,comp(X)) \vee mem(E,X) $
\item
$ \neg mem(E,X) \vee \neg mem(E,comp(X)) $

\item
$ \neg both(X,Y,S) \vee mem(X,S) $
\item
$ \neg both(X,Y,S) \vee mem(Y,S) $
\item
$ \neg mem(X,S) \vee \neg m(Y,S) \vee both(X,Y,S) $
\item
$ \neg distinct(X,Y) \vee \neg both(X,Y,a) $
\item
$ \neg distinct(X,Y) \vee \neg both(X,Y,b) $
\item
$ \neg distinct(X,Y) \vee \neg both(X,Y,c) $
\item
$ both(X,Y,a) \vee both(X,Y,b) \vee both(X,Y,c) \vee distinct(X,Y) $
\item
$ \neg number(X) \vee mem(X,union(union(a,b),c)) $
\item
$ \neg triple(X,Y,Z) \vee distinct(X,Y) \vee distinct(X,Z) \vee distinct(Y,Z) $
\end{enumerate}

\end{document}